\documentclass[a4paper,fleqn,11pt]{article}

\usepackage{amsmath}
\usepackage{latexsym}
\usepackage{amssymb}
\usepackage{graphicx}
\usepackage{amsthm}

\newtheorem{thm}{Theorem}[section]
\newtheorem{lem}[thm]{Lemma}
\newtheorem{cor}[thm]{Corollary}

\theoremstyle{definition}
\newtheorem{defn}{Definition}[section]

\begin{document}
\title{Constructing Two Edge-Disjoint Hamiltonian Cycles in Locally Twisted Cubes}
\author{\vspace{0.5cm}Ruo-Wei Hung\\
Department of Computer Science and Information Engineering,\\
Chaoyang University of Technology,\\
Wufong, Taichung 41349, Taiwan}

\maketitle










\begin{abstract}
The $n$-dimensional hypercube network $Q_n$ is one of the most
popular interconnection networks since it has simple structure and
is easy to implement. The $n$-dimensional locally twisted cube,
denoted by $LTQ_n$, an important variation of the hypercube, has
the same number of nodes and the same number of connections per
node as $Q_n$. One advantage of $LTQ_n$ is that the diameter is
only about half of the diameter of $Q_n$. Recently, some
interesting properties of $LTQ_n$ were investigated. In this
paper, we construct two edge-disjoint Hamiltonian cycles in the
locally twisted cube $LTQ_n$, for any integer $n\geqslant 4$. The
presence of two edge-disjoint Hamiltonian cycles provides an
advantage when implementing algorithms that require a ring
structure by allowing message traffic to be spread evenly across
the locally twisted cube.\\

\end{abstract}






\section{Introduction}
Parallel computing is important for speeding up computation. The
design of an interconnection network is the first thing to be
considered. Many topologies have been proposed in the literature
\cite{Bhuyan84,Cull95,Efe92,Hilbers87}, and the desirable
properties of an interconnection network include symmetry,
relatively small degree, small diameter, embedding capabilities,
scalability, robustness, and efficient routing. Among those
proposed interconnection networks, the hypercube is a popular
interconnection network with many attractive properties such as
regularity, symmetry, small diameter, strong connectivity,
recursive construction, partition ability, and relatively low link
complexity \cite{Saad88}. The architecture of an interconnection
network is usually modeled by a graph, where the nodes represent
the processing elements and the edges represent the communication
links. In this paper, we will use graphs and networks
interchangeably.

The $n$-dimensional locally twisted cube, denoted by $LTQ_n$, was
first proposed by Yang et al. \cite{Yang04,Yang05} and is a better
hypercube variant which is conceptually closer to the comparable
hypercube $Q_n$ than existing variants. The $n$-dimensional
locally twisted cube $LTQ_n$ is similar to $n$-dimensional
hypercube $Q_n$ in the sense that the nodes can be one-to-one
labeled with 0-1 binary strings of length $n$, so that the labels
of any two adjacent nodes differ in at most two successive bits.
One advantage is that the diameter of locally twisted cubes is
only about half the diameter of hypercubes \cite{Yang05}.
Recently, some interesting properties, such as conditional link
faults, of the locally twisted cube $LTQ_n$ were investigated.
Yang et al. proved that $LTQ_n$ has a connectivity of $n$
\cite{Yang05}. They also showed that locally twisted cubes are
4-pancyclic and that a locally twisted cube is superior to a
hypercube in terms of ring embedding capability \cite{Yang04}. Ma
and Xu \cite{Ma06} showed that for any two different nodes $u$ and
$v$ in $LTQ_n$ ($n \geqslant 3$), there exists a $uv$-path of
length $l$ with $d(u,v)+2 \leqslant l \leqslant 2^n-1$ except for
a shortest $uv$-path, where $d(u,v)$ is the length of a shortest
path between $u$ and $v$. In \cite{Yang07}, Yang et al. addressed
the fault diagnosis of locally twisted cubes under the $MM^*$
comparison model. Hsieh et al. constructed $n$ edge-disjoint
spanning trees in an $n$-dimensional locally twisted cube
\cite{Hsieh09}. Recently, Hsieh et al. showed that for any $LTQ_n$
($n\geqslant 3$) with at most $2n-5$ faulty edges in which each
node is incident to at least two fault-free edges, there exists a
fault-free Hamiltonian cycle \cite{Hsieh10}.

Two Hamiltonian cycles in a graph are said to be
\textit{edge-disjoint} if they do not share any common edge. The
edge-disjoint Hamiltonian cycles can provide advantage for
algorithms that make use of a ring structure \cite{Rowley91}. The
following application about edge-disjoint Hamiltonian cycles can
be found in \cite{Rowley91}. Consider the problem of all-to-all
broadcasting in which each node sends an identical message to all
other nodes in the network. There is a simple solution for the
problem using an $n$-node ring that requires $n-1$ steps, i.e., at
each step, every node receives a new message from its ring
predecessor and passes the previous message to its ring successor.
If the network admits edge-disjoint rings, then messages can be
divided and the parts broadcast along different rings without any
edge contention. If the network can be decomposed into
edge-disjoint Hamiltonian cycles, then the message traffic will be
evenly distributed across all communication links. Edge-disjoint
Hamiltonian cycles also form the basis of an efficient all-to-all
broadcasting algorithm for networks that employ warmhole or
cut-through routing \cite{Lee90}.

The edge-disjoint Hamiltonian cycles in $k$-ary $n$-cubes and
hypercubes has been constructed in \cite{Bae03}. Barden et al.
constructed the maximum number of edge-disjoint spanning trees in
a hypercube \cite{Barden99}. Petrovic et al. characterized the
number of edge-disjoint Hamiltonian cycles in hyper-tournaments
\cite{Petrovic06}. Hsieh et al. constructed edge-disjoint spanning
trees in locally twisted cubes \cite{Hsieh09}. Hsieh et al.
investigated the edge-fault tolerant Hamiltonicity of an
$n$-dimensional locally twisted cube \cite{Hsieh10}. The existence
of a Hamiltonian cycle in locally twisted cubes has been verified
\cite{Yang04}. However, there has been little work reported so far
on edge-disjoint properties in the locally twisted cubes. In this
paper, we show that, for any integer $n\geqslant 4$, there are two
edge-disjoint Hamiltonian cycles in the $n$-dimensional locally
twisted cube $LTQ_n$.

The rest of the paper is organized as follows. In Section
\ref{Preliminaries}, the structure of the locally twisted cube is
introduced, and some definitions and notations used in this paper
are given. Section \ref{EdgeDisjointHC} shows the construction of
two edge-disjoint Hamiltonian cycles in the locally twisted cube.
Finally, we conclude this paper in Section \ref{Conclusion}.

\section{Preliminaries}\label{Preliminaries}
We usually use a graph to represent the topology of an
interconnection network. A graph $G = (V, E)$ is a pair of the
node set $V$ and the edge set $E$, where $V$ is a finite set and
$E$ is a subset of $\{(u,v)| (u,v)$ is an unordered pair of $V\}$.
We will use $V(G)$ and $E(G)$ to denote the node set and the edge
set of $G$, respectively. If $(u,v)$ is an edge in a graph $G$, we
say that $u$ \textit{is adjacent to} $v$. A \textit{neighbor} of a
node $v$ in a graph $G$ is any node that is adjacent to $v$.
Moreover, we use $N_G(v)$ to denote the neighbors of $v$ in $G$.
The subscript `$G$' of $N_G(v)$ can be removed from the notation
if it has no ambiguity.

A path $P$, represented by $\langle v_0\rightarrow  v_1\rightarrow
\cdots\rightarrow v_{t-1} \rangle$, is a sequence of distinct
nodes such that two consecutive nodes are adjacent. The first node
$v_0$ and the last node $v_{t-1}$ visited by $P$ are called the
\textit{path-start} and \textit{path-end} of $P$, denoted by
$start(P)$ and $end(P)$, respectively, and they are called the
\textit{end nodes} of $P$. Path $\langle v_{t-1}\rightarrow \cdots
\rightarrow v_1 \rightarrow v_0 \rangle$ is called the
\textit{reversed path}, denoted by $P_{\textrm{rev}}$, of $P$.
That is, $P_{\textrm{rev}}$ visits the vertices of $P$ from
$end(P)$ to $start(P)$ sequently. In addition, $P$ is a cycle if
$|V(P)|\geqslant 3$ and $end(P)$ is adjacent to $start(P)$. A path
$\langle v_0\rightarrow v_1\rightarrow \cdots\rightarrow
v_{t-1}\rangle$ may contain other subpath $Q$, denoted as $\langle
v_0\rightarrow v_1\rightarrow \cdots\rightarrow v_i\rightarrow
Q\rightarrow v_j\cdots\rightarrow v_{t-1}\rangle$, where
$Q=\langle v_{i+1}\rightarrow v_{i+2}\rightarrow \cdots\rightarrow
v_{j-1} \rangle$. A path (or cycle) in $G$ is called a
\textit{Hamiltonian path} (or \textit{Hamiltonian cycle}) if it
contains every node of $G$ exactly once. Two paths (or cycles)
$P_1$ and $P_2$ connecting a node $u$ to a node $v$ are said to be
\textit{edge-disjoint} iff $E(P_1)\cap E(P_2)=\emptyset$. Two
paths (or cycles) $Q_1$ and $Q_2$ of graph $G$ are called
\textit{node-disjoint} iff $V(Q_1)\cap V(Q_2)=\emptyset$. Two
node-disjoint paths $Q_1$ and $Q_2$ can be \textit{concatenated}
into a path, denoted by $Q_1\Rightarrow Q_2$, if $end(Q_1)$ is
adjacent to $start(Q_2)$.

Now, we introduce locally twisted cubes. A node of the
$n$-dimensional locally twisted cube $LTQ_n$ is represented by a
0-1 binary string of length $n$. A binary string $b$ of length $n$
is denoted by $b_{n-1}b_{n-2}\cdots b_1b_0$, where $b_{n-1}$ is
the most significant bit. We then give the recursive definition of
the $n$-dimensional locally twisted cube $LTQ_n$, for any integer
$n\geqslant 2$, as follows.

\begin{defn}\cite{Yang04,Yang05}\label{def_locally-twisted-cube}
Let $n \geqslant 2$. The $n$-dimensional locally twisted cube,
denoted by $LTQ_n$, is defined recursively as follows.\\
(1) $LTQ_2$ is a graph consisting of four nodes labeled with 00,
01, 10, and 11, respectively, connected by four edges (00, 01),
(00, 10), (01, 11), and (10, 11).\\
(2) For $n\geqslant 3$, $LTQ_n$ is built from two disjoint copies
$LTQ_{n-1}$ according to the following steps. Let $LTQ_{n-1}^0$
denote the graph obtained by prefixing the label of each node of
one copy of $LTQ_{n-1}$ with $0$, let $LTQ_{n-1}^1$ denote the
graph obtained by prefixing the label of each node of the other
copy of $LTQ_{n-1}$ with $1$, and connect each node
$b=0b_{n-2}b_{n-3}\cdots b_1b_0$ of $LTQ_{n-1}^0$ with the node
$1(b_{n-2}\oplus b_0)b_{n-3}\cdots b_1b_0$ of $LTQ_{n-1}^1$ by an
edge, where `$\oplus$' represents the modulo 2 addition.
\end{defn}

According to Definition \ref{def_locally-twisted-cube}, $LTQ_n$ is
an $n$-regular graph with $2^n$ nodes and $n2^{n-1}$ edges. Note
that $0\oplus 0 = 1\oplus 1 = 0$ and $0\oplus 1 = 1\oplus 0 = 1$.
The $n$-dimensional locally twisted cube $LTQ_n$ is closed to an
$n$-dimensional hypercube $Q_n$ except that the labels of any two
adjacent nodes in $LTQ_n$ differ in at most two successive bits.
In addition, $LTQ_n$ can be decomposed into two sub-locally
twisted cubes $LTQ_{n-1}^{0}$ and $LTQ_{n-1}^{1}$, where for each
$i\in\{0,1\}$, $LTQ_{n-1}^i$ consists of those nodes
$b=b_{n-1}b_{n-2}\cdots b_1b_0$ with $b_{n-1}=i$. For each
$i\in\{0,1\}$, $LTQ_{n-1}^i$ is isomorphic to $LTQ_{n-1}$. For
example, Fig. \ref{Fig_LTQ_34}(a) shows $LTQ_3$ and Fig.
\ref{Fig_LTQ_34}(b) depicts $LTQ_4$ containing two sub-locally
twisted cubes $LTQ_3^0$ and $TQ_3^1$.

\begin{figure}[t]
\begin{center}
\includegraphics[scale=0.7]{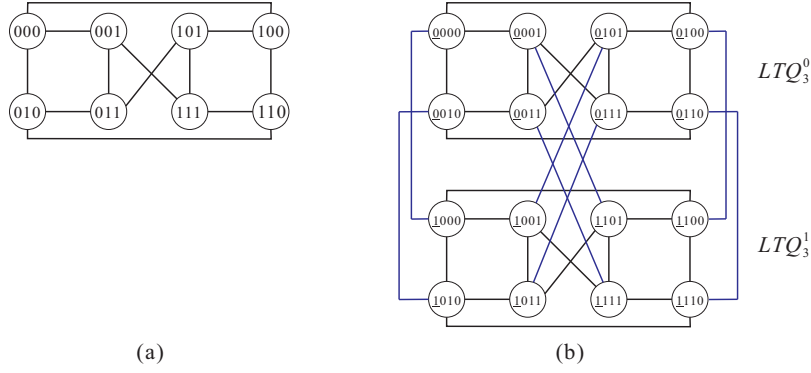}
\caption{(a) The 3-dimensional locally twisted cube $LTQ_3$, and
(b) the 4-dimensional locally twisted cube $LTQ_4$ containing
sub-locally twisted cubes $LTQ_3^0$ and $LTQ_3^1$.}
\label{Fig_LTQ_34}
\end{center}
\end{figure}

Let $b$ is a binary string $b_{t-1}b_{t-2}\cdots b_1b_0$ of length
$t$. We denote $b^i$ the new binary string obtained by repeating
$b$ string $i$ times. For instance, $(10)^2=1010$ and $0^3=000$.

\section{Two Edge-Disjoint Hamiltonian Cycles}\label{EdgeDisjointHC}
Obviously, $LTQ_3$ has no two edge-disjoint Hamiltonian cycles
since each node is incident to three edges. Our method for
constructing two edge-disjoint Hamiltonian cycles of $LTQ_n$, with
integer $n\geqslant 4$, is based on an inductive construction.
Initially, we construct two edge-disjoint Hamiltonian paths, $P$
and $Q$, of $LTQ_4$ so that $start(P)=0010$, $end(P)=0000$,
$start(Q)=0110$, and $end(Q)=0100$. Clearly, these two paths are
two edge-disjoint Hamiltonian cycles of $LTQ_4$. For $n\geqslant
5$, we will construct two edge-disjoint Hamiltonian paths $P$ and
$Q$ in $LTQ_n$ such that $start(P)=00(0)^{n-5}010$,
$end(P)=10(0)^{n-5}010$, $start(Q)=00(0)^{n-5}110$, and
$end(Q)=10(0)^{n-5}110$. By Definition
\ref{def_locally-twisted-cube}, $start(P)\in N(end(P))$ and
$start(Q)\in N(end(Q))$. Thus, $P$ and $Q$ are two edge-disjoint
Hamiltonian cycles.

We first show that $LTQ_4$ contains two edge-disjoint Hamiltonian
paths in the following lemma.

\begin{lem}\label{2HP-LTQ_4}
There are two edge-disjoint Hamiltonian paths $P$ and $Q$ in
$LTQ_4$ such that $start(P)=0010$, $end(P)=0000$, $start(Q)=0110$,
and $end(Q)=0100$.
\end{lem}
\begin{proof}
We prove this lemma by constructing such two paths. Let\\
$P=\langle$0010 $\rightarrow$ 0110 $\rightarrow$ 0111
$\rightarrow$ 0101 $\rightarrow$ 0100 $\rightarrow$ 1100
$\rightarrow$ 1110 $\rightarrow$ 1010 $\rightarrow$ 1000
$\rightarrow$ 1001 $\rightarrow$ 1011 $\rightarrow$ 1101
$\rightarrow$ 1111 $\rightarrow$ 0011 $\rightarrow$ 0001
$\rightarrow$ 0000$\rangle$, and let\\
$Q=\langle$0110 $\rightarrow$ 1110 $\rightarrow$ 1111
$\rightarrow$ 1001 $\rightarrow$ 0101 $\rightarrow$ 0011
$\rightarrow$ 0010 $\rightarrow$ 1010 $\rightarrow$ 1011
$\rightarrow$ 0111 $\rightarrow$ 0001 $\rightarrow$ 1101
$\rightarrow$ 1100 $\rightarrow$ 1000 $\rightarrow$ 0000
$\rightarrow$ 0100$\rangle$.\\
Fig. \ref{Fig_LTQ_4-2HP} depicts the constructions of $P$ and $Q$.
Clearly, $P$ and $Q$ are edge-disjoint Hamiltonian paths in
$LTQ_4$.
\end{proof}

\begin{figure}[t]
\begin{center}
\includegraphics[scale=0.75]{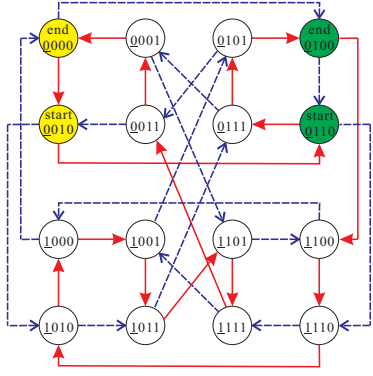}
\caption{Two edge-disjoint Hamiltonian paths in $LTQ_4$, where
solid arrow lines indicate a Hamiltonian path $P$ and dotted arrow
lines indicate the other edge-disjoint Hamiltonian path $Q$.}
\label{Fig_LTQ_4-2HP}
\end{center}
\end{figure}

According to Definition \ref{def_locally-twisted-cube}, nodes 0010
and 0000 are adjacent, and nodes 0110 and 0100 are adjacent. Thus,
the following corollary immediately holds true from Lemma
\ref{2HP-LTQ_4}.

\begin{cor}\label{2HC-LTQ_4}
There are two edge-disjoint Hamiltonian cycles in $LTQ_4$.
\end{cor}

Using Lemma \ref{2HP-LTQ_4}, we show that $LTQ_5$ has two
edge-disjoint Hamiltonian paths in the following lemma.

\begin{lem}\label{2HP-LTQ_5}
There are two edge-disjoint Hamiltonian paths $P$ and $Q$ in
$LTQ_5$ such that $start(P)=00010$, $end(P)=10010$,
$start(Q)=00110$, and $end(Q)=10110$.
\end{lem}
\begin{proof}
We first partition $LTQ_5$ into two sub-locally twisted cubes
$LTQ_4^0$ and $LTQ_4^1$. By Lemma \ref{2HP-LTQ_4}, there are two
edge-disjoint Hamiltonian paths $P^i$ and $Q^i$ in $LTQ_4^i$, for
$i\in\{0,1\}$, such that $start(P^i)=i0010$, $end(P^i)=i0000$,
$start(Q^i)=i0110$, and $end(Q^i)=i0100$. By Definition
\ref{def_locally-twisted-cube}, we have that $end(P^0)\in
N(end(P^1))$ and $end(Q^0)\in N(end(Q^1))$.\\
Let $P=P^0 \Rightarrow P_{\textrm{rev}}^1$ and let $Q=Q^0
\Rightarrow Q_{\textrm{rev}}^1$, where $P_{\textrm{rev}}^1$ and
$Q_{\textrm{rev}}^1$ are the reversed paths of $P^1$ and $Q^1$,
respectively. Then, $P$ and $Q$ are two edge-disjoint Hamiltonian
paths in $LTQ_5$ such that $start(P)=00010$, $end(P)=10010$,
$start(Q)=00110$, and $end(Q)=10110$. Fig. \ref{Fig_LTQ_5-2HP}
shows the constructions of such two edge-disjoint Hamiltonian
paths in $LTQ_5$. Thus, the lemma hods true.
\end{proof}

\begin{figure}[t]
\begin{center}
\includegraphics[scale=0.73]{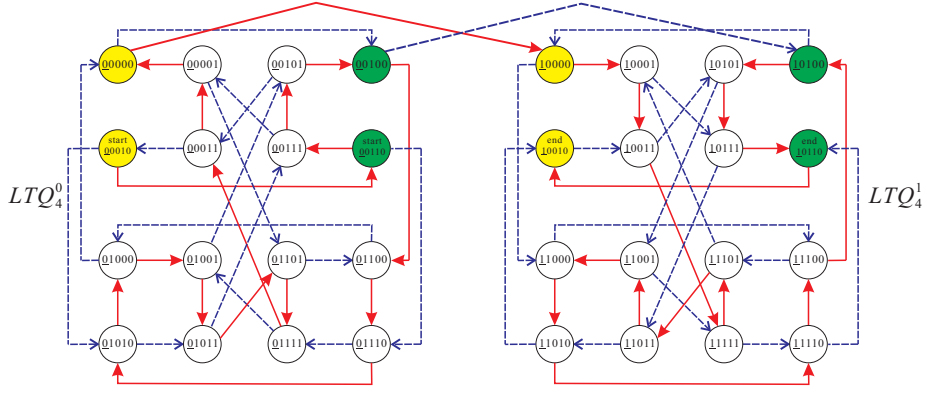}
\caption{Two edge-disjoint Hamiltonian paths in $LTQ_5$, where
solid arrow lines indicate a Hamiltonian path $P$ and dotted arrow
lines indicate the other edge-disjoint Hamiltonian path $Q$.}
\label{Fig_LTQ_5-2HP}
\end{center}
\end{figure}

According to Definition \ref{def_locally-twisted-cube}, nodes
00010 and 10010 are adjacent, and nodes 00110 and 10110 are
adjacent. Thus, the following corollary immediately holds true
from Lemma \ref{2HP-LTQ_5}.

\begin{cor}\label{2HC-LTQ_5}
There are two edge-disjoint Hamiltonian cycles in $LTQ_5$.
\end{cor}

Based on Lemma \ref{2HP-LTQ_5}, we prove the following lemma.

\begin{lem}\label{2HP-LTQ}
For any integer $n\geqslant 5$, there are two edge-disjoint
Hamiltonian paths $P$ and $Q$ in $LTQ_n$ such that
$start(P)=00(0)^{n-5}010$, $end(P)=10(0)^{n-5}010$,
$start(Q)=00(0)^{n-5}110$, and $end(Q)=10(0)^{n-5}110$.
\end{lem}
\begin{proof}
We prove this lemma by induction on $n$, the dimension of the
locally twisted cube. It follows from Lemma \ref{2HP-LTQ_5} that
the lemma holds for $n=5$. Suppose that the lemma is true for the
case $n=k$ ($k\geqslant 5$). Assume that $n=k+1$. We first
partition $LTQ_{k+1}$ into two sub-locally twisted cubes
$LTQ_{k}^0$ and $LTQ_{k}^1$. By the induction hypothesis, there
are two edge-disjoint Hamiltonian paths $P^i$ and $Q^i$ in
$LTQ_k^i$, for $i\in\{0,1\}$, such that
$start(P^i)=i00(0)^{k-5}010$, $end(P^i)=i10(0)^{k-5}010$,
$start(Q^i)=i00(0)^{k-5}110$, and $end(Q^i)=i10(0)^{k-5}110$. By
Definition \ref{def_locally-twisted-cube}, we have that
$end(P^0)\in N(end(P^1))$ and $end(Q^0)\in N(end(Q^1))$.\\
Let $P=P^0 \Rightarrow P_{\textrm{rev}}^1$ and let $Q=Q^0
\Rightarrow Q_{\textrm{rev}}^1$, where $P_{\textrm{rev}}^1$ and
$Q_{\textrm{rev}}^1$ are the reversed paths of $P^1$ and $Q^1$,
respectively. Then, $P$ and $Q$ are two edge-disjoint Hamiltonian
paths in $LTQ_{k+1}$ such that $start(P)=00(0)^{k-4}010$,
$end(P)=10(0)^{k-4}010$, $start(Q)=00(0)^{k-4}110$, and
$end(Q)=10(0)^{k-4}110$. Fig. \ref{Fig_LTQ_k-2HP} depicts the
constructions of such two edge-disjoint Hamiltonian paths in
$LTQ_{k+1}$. Thus, the lemma hods true when $n=k+1$. By induction,
the lemma holds true.
\end{proof}

\begin{figure}[t]
\begin{center}
\includegraphics[scale=0.75]{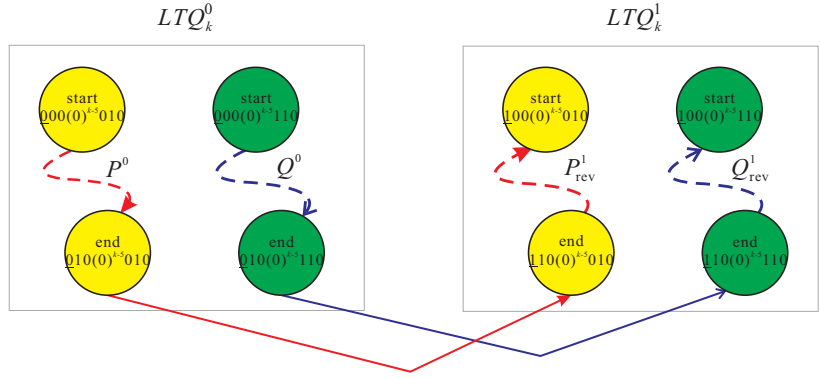}
\caption{The constructions of two edge-disjoint Hamiltonian paths
in $LTQ_{k+1}$, with $k\geqslant 5$, where dotted arrow lines
indicate the paths and solid arrow lines indicate concatenated
edges.} \label{Fig_LTQ_k-2HP}
\end{center}
\end{figure}

By Definition \ref{def_locally-twisted-cube}, nodes
$start(P)=00(0)^{n-5}010$ and $end(P)=10(0)^{n-5}010$ are
adjacent, and nodes $start(Q)=00(0)^{n-5}110$ and
$end(Q)=10(0)^{n-5}110$ are adjacent. It immediately follows from
Lemma \ref{2HP-LTQ} that the following corollary holds true.

\begin{cor}\label{2HC-LTQ}
For any integer $n\geqslant 5$, there are two edge-disjoint
Hamiltonian cycles in $LTQ_n$.
\end{cor}

It immediately follows from Lemmas \ref{2HP-LTQ_4} and
\ref{2HP-LTQ}, and Corollaries \ref{2HC-LTQ_4} and \ref{2HC-LTQ},
that the following two theorems hold true.

\begin{thm}\label{2HP-LTQ_theorem}
For any integer $n\geqslant 4$, there are two edge-disjoint
Hamiltonian paths in $LTQ_n$.
\end{thm}

\begin{thm}\label{2HC-LTQ_theorem}
For any integer $n\geqslant 4$, there are two edge-disjoint
Hamiltonian cycles in $LTQ_n$.
\end{thm}

\section{Concluding Remarks}\label{Conclusion}
In this paper, we construct two edge-disjoint Hamiltonian cycles
(paths) in a $n$-dimensional locally twisted cubes $LTQ_n$, for
any integer $n\geqslant 4$. In the construction of two
edge-disjoint Hamiltonian cycles (paths) of $LTQ_n$, some edges
are not used. It is interesting to see if there are more
edge-disjoint Hamiltonian cycles of $LTQ_n$ for $n\geqslant 6$. We
would like to post it as an open problem to interested readers.



\begin{thebibliography}{99}

\bibitem{Bae03}
M.M. Bae, B. Bose, Edge disjoint Hamiltonian cycles in $k$-ary
$n$-cubes and hypercubes, IEEE Trans. Comput. 52 (2003)
1271--1284.

\bibitem{Barden99}
B. Barden, R. Libeskind-Hadas, J. Davis, W. Williams, On
edge-disjoint spanning trees in hypercubes, Inform. Process. Lett.
70 (1999) 13--16.

\bibitem{Bhuyan84}
L.N. Bhuyan, D.P. Agrawal, Generalized hypercube and hyperbus
structures for a computer network, IEEE Trans. Comput. C-33 (1984)
323--333.

\bibitem{Cull95}
P. Cull, S.M. Larson, The M\"{o}bius cubes, IEEE Trans. Comput. 44
(1995) 647--659.

\bibitem{Efe92}
K. Efe, The crossed cube architecture for parallel computing, IEEE
Trans. Parallel Distribut. Syst. 3 (1992) 513--524.

\bibitem{Hilbers87}
P.A.J. Hilbers, M.R.J. Koopman, J.L.A. van de Snepscheut, The
twisted cube, in: J. deBakker, A. Numan, P. Trelearen (Eds.),
PARLE: Parallel Architectures and Languages Europe, Parallel
Architectures, vol. 1, Springer, Berlin, 1987, pp. 152--158.

\bibitem{Hsieh09}
S.Y. Hsieh, C.J. Tu, Constructing edge-disjoint spanning trees in
locally twisted cubes, Theoret. Comput. Sci. 410 (2009) 926--932.

\bibitem{Hsieh10}
S.Y. Hsieh, C.Y. Wu, Edge-fault-tolerant Hamiltonicity of locally
twisted cubes under conditional edge faults, J. Comb. Optim. 19
(2010) 16--30.

\bibitem{Lee90}
S. Lee, K.G. Shin, Interleaved all-to-all reliable broadcast on
meshes and hypercubes, in: Proc. Int. Conf. Parallel Processing,
vol. 3, 1990, pp. 110--113.

\bibitem{Ma06}
M. Ma, J.M. Xu, Panconnectivity of locally twisted cubes, Appl.
Math. Lett. 19 (2006) 673--677.

\bibitem{Petrovic06}
V. Petrovic, C. Thomassen, Edge-disjoint Hamiltonian cycles in
hypertournaments, J. Graph Theory 51 (2006) 49--52.

\bibitem{Saad88}
Y. Saad, M.H. Schultz, Topological properties of hypercubes, IEEE
Trans. Comput. 37 (1988) 867--872.

\bibitem{Rowley91}
R. Rowley, B. Bose, Edge-disjoint Hamiltonian cycles in de Bruijn
networks, in: Proc. 6th Distributed Memory Computing Conference,
1991, pp. 707--709.

\bibitem{Yang07}
H. Yang, X. Yang, A fast diagnosis algorithm for locally twisted
cube multiprocessor systems under the $MM^*$ model, Comput. Math.
Appl. 53 (2007) 91--926.

\bibitem{Yang04}
X. Yang, G.M. Megson, D.J. Evans, Locally twisted cubes are
4-pancyclic, Appl. Math. Lett. 17 (2004) 919--925.

\bibitem{Yang05}
X. Yang, D.J. Evans, G.M. Megson, The locally twisted cubes, Int.
J. Comput. Math. 82 (2005) 401--413.

\end{thebibliography}
\end{document}